\documentclass[prl,aps,twocolumn,superscriptaddress,nofootinbib,preprintnumbers]{revtex4-2}
\let\phi\varphi

\usepackage[dvipsnames]{xcolor}
\usepackage{lipsum}
\usepackage{amsmath}
\usepackage{amsfonts}
\usepackage{amsthm}
\usepackage{amssymb}
\usepackage{cancel}
\usepackage{relsize}
\usepackage[utf8]{inputenc}
\usepackage{graphicx}
\usepackage[T1]{fontenc}
\usepackage{setspace}
\usepackage[ colorlinks = true,
             linkcolor = MidnightBlue,
             urlcolor  = MidnightBlue,
             citecolor = orange,
             anchorcolor = ForestGreen,
]{hyperref}
\usepackage[capitalize,nameinlink]{cleveref}
\AddToHook{cmd/appendix/before}{\crefalias{section}{appendix}}
\usepackage[english]{babel}
\usepackage{pifont}
\usepackage{placeins}
\usepackage[font=small,labelfont=bf]{caption}
\usepackage{url}
\usepackage[shortcuts]{extdash}
\usepackage{tikz}
\usetikzlibrary{arrows.meta}
\usetikzlibrary{decorations.pathmorphing}
\usetikzlibrary{shapes}
\tikzstyle{causalarrow} = [draw=black,thick,decorate,decoration={snake,amplitude=.7mm,pre length=1mm,post length=2mm}]
\tikzstyle{affectsarrow} = [draw=black,thick]
\tikzstyle{vvarrow} = [draw=black,thick,decorate,decoration={zigzag,amplitude=+1.3pt,segment length=+7pt,pre length=1.5pt,post length=+7pt}]
\tikzstyle{routedarrow} = [draw=black,very thick]
\usetikzlibrary{decorations.markings}
\tikzset{
  strike through/.style={
    postaction=decorate,
    decoration={
      markings,
      mark=at position 0.5 with {
        \draw[-] (-3pt,-3pt) -- (3pt, 3pt);
      }
    }
  }
}
\tikzset{
  defcolon/.style={
    postaction=decorate,
    decoration={
      markings,
      mark=at position -0.05 with {$\colon$}
    }
  }
}
\usepackage{tikz-3dplot}
\tdplotsetmaincoords{80}{45}
\tdplotsetrotatedcoords{-90}{180}{-90}
\usepackage{verbatim} % Add comment environment
\usepackage{bbold} % makes $\mathbb{1}$ work
\usepackage{braket} % brackets
\usepackage{csquotes} % biblatex wanted this

\usepackage{stmaryrd} % lightning for math!
\SetSymbolFont{stmry}{bold}{U}{stmry}{m}{n} % lightning for math yields warnings...
\usepackage[header]{appendix}
\usepackage{mathtools} % provides mathmakebox
\usepackage{mathrsfs} % provides mathscr
\usepackage{booktabs} % midrule
\usepackage{mdframed}
\usepackage{array}
\usepackage{bm}

\usepackage{xstring}
\usepackage{subcaption}
\usepackage{multirow}

\graphicspath{
	{../tikz/}
	{../pics/}
    {pics/}
}

\let\AA\relax
\newcommand{\AA}{\ensuremath{\mathcal{A}}}
\newcommand{\BB}{\ensuremath{\mathcal{B}}}

\newcommand{\TT}{\ensuremath{\mathcal{T}}}
\newcommand{\WW}{\ensuremath{\mathcal{W}}}
\newcommand{\XX}{\ensuremath{\mathcal{X}}}
\newcommand{\YY}{\ensuremath{\mathcal{Y}}}
\newcommand{\ZZ}{\ensuremath{\mathcal{Z}}}
\newcommand{\SC}{\ensuremath{\mathcal{S}}}

\newcommand{\Fut}{\bar{\mathcal{F}}}

\newcommand{\abs}[1]{\left| #1 \right|}

\newcommand{\unord}{\ensuremath{\not\preceq \not\succeq}}

\newcommand{\editorial}[1]{}

\newcommand{\cG}{\mathcal{G}}

\DeclareMathOperator{\spann}{span}

\DeclareMathOperator{\doo}{do}

\DeclareMathOperator{\affects}{\mathrel{\vDash}}
\DeclareMathOperator{\naffects}{\not\affects}
\DeclareMathOperator{\cause}{\twoheadrightarrow}

\DeclareMathOperator{\given}{\,|\,}

\usepackage{scalerel}

\usepackage{thm-restate}
\makeatletter
\newcommand{\IfRestatedTF}[2]{\ifthmt@thisistheone #2\else #1\fi}
\makeatother
\usepackage{thmtools} 

\newtheoremstyle{note}%hnamei
{1.2em}                %   Space above
{1.2em}                %   Space below
{}                     %   Body font
{}                     %   Indent amount
{\bfseries}    %   Theorem head font
{.}                    %   Punctuation after theorem head
{.5em}                 %   Space after theorem head
{}                     %   Theorem head spec (can be left empty, meaning `normal')
\theoremstyle{note}

\newtheorem{theorem}{Theorem}[section]
\newtheorem{definition}[theorem]{Definition}

\newtheorem{lemma}[theorem]{Lemma}
\newtheorem{corollary}[theorem]{Corollary}

\newtheorem{remark}[theorem]{Remark}

\usepackage{centernot}
\begin{document}
\author{Maarten Grothus}
    \affiliation{Univ.\ Grenoble Alpes, Inria, 38000 Grenoble, France}
    \affiliation{Univ.\ Grenoble Alpes, CNRS, Grenoble INP, Institut N\'eel, 38000 Grenoble, France}
\author{V. Vilasini}
    \affiliation{Univ.\ Grenoble Alpes, Inria, 38000 Grenoble, France}
    \affiliation{Institute for Theoretical Physics, ETH Zürich, 8093 Zürich, Switzerland}
\title{
    Impossibility of superluminal signalling rules out causal loops in conical spacetimes}

\begin{abstract}
In PRL 129, 110401 it was shown that it is theoretically possible to have operationally detectable causal loops without violating the principle of no superluminal signalling (NSS) in ($1+1$)-Minkowski spacetime.
Whether or not such causal loops are also possible in $d > 1$ spatial dimensions, has remained a key open question.
We resolve this question by showing that in a wide class of \enquote{conical} spacetimes, including Minkowski with $d > 1$, NSS does rule out all operationally detectable causal loops, in classical, quantum and post-quantum theories.
This establishes that the relationship between the relativistic principles of NSS and no causal loops depends inherently on the geometry of spacetime.
\end{abstract}

\maketitle

The interface between information-theoretic and spatiotemporal notions of causality has attracted growing interest in recent years, motivated by the recognition that these two conceptions of causal order are markedly distinct and yet must interact consistently in physical scenarios. On the spatiotemporal side, exotic causal structures such as closed timelike curves (CTCs) in general relativity have often been considered unphysical. On the information-theoretic side, by contrast, a rich landscape of cyclic causal structures---both classical and non-classical---has been developed and garnered interest for their information-processing consequences, extending well beyond the fixed acyclic causal models of statistics and quantum foundations (e.g., \cite{Deutsch1991,Lloyd2011,Oreshkov2012,Chiribella2013,Lugano2014}). Reconciling these perspectives raises a fundamental question: when an information-theoretic causal structure is embedded in an acyclic (CTC-free) spacetime consistently with relativistic causality principles, must it itself be acyclic?

This question was first considered in~\cite{VVC, VVC_Letter} by imposing relativistic principles such as no-superluminal signalling (NSS) which relate the two types of causal structures. Popular belief, backed by special relativity, suggests that NSS in Minkowski spacetime should rule out causal loops. Surprisingly, this turns out to be false in general. Within an axiomatic, theory-independent formalism (the \emph{affects framework})~\cite{VVC}, built on causal modelling~\cite{Spirtes1993,Pearl2009}, it was shown in~\cite{VVC_Letter} that there exist hypothetical theories in which NSS in $(1+1)$-Minkowski spacetime can consistently coexist with causal loops that are nonetheless \emph{operationally verifiable} through their observable signalling relations (formally called \emph{affects causal loops}, ACLs). The underlying reason is that causation and signalling, often conflated, are inequivalent in the presence of fine-tuning, where causal mechanisms are carefully tuned to hide observable signalling. As a consequence, NSS does not, in general, imply no superluminal causation~\cite{VVCPR}, and~\cite{VVC_Letter} exploited this in $(1+1)$-Minkowski to construct causal loops consistent with NSS. 

Whether the same is true in $(d+1)$-Minkowski spacetimes with $d>1$ has remained a key open problem, which we address in this work. While a negative answer was conjectured in~\cite{VVC,VVC_Letter}, a significant challenge in proving it comes from the fact that the affects framework admits a vast landscape of possible ACLs involving arbitrary numbers of systems, whose verifiability is certified through complex patterns of signalling relations (\emph{higher-order affects relations}); ruling them all out requires controlling this entire set of ACLs at once. Moreover, the precise geometric feature of spacetime responsible for the $d=1$ versus $d>1$ distinction was, until recently, unclear. Our recent work~\cite{Grothus2024} introduced tools to characterise higher-order affects relations and identified a precise order-theoretic property of spacetime---\emph{conicality}---that  holds in $(d+1)$-Minkowski for $d>1$ but fails for $d=1$. Beyond that, \cite{Paganini2026} shows that conicality holds for a wide class of causally regular spacetimes.
These works however still left open the main question considered here, namely whether these tools are sufficient to rule out the entire landscape of ACLs in conical spacetimes. 

Here, we close this gap definitively, building on~\cite{VVC_Letter,Grothus2024} to prove a general no-go result: in \emph{any} conical spacetime, NSS rules out \emph{all} operationally verifiable causal loops (or ACLs), in contrast to the situation in non-conical spacetimes~\cite{VVC_Letter}. This establishes that the relationship between the relativistic principles of NSS and no causal loops (NCL) depends on the geometry of spacetime.
These results may therefore bear relevance for the larger research program on understanding the interface of quantum information theory and spacetime structure.

\section{Brief overview of the affects framework and conicality}
\label{sec:overview}
{\bf Causal models, signalling and causal inference} The \emph{affects framework}~\cite{VVC} first provides a general definition of causal models, consisting of (1) a directed graph $\mathcal{G}$ with a subset $S = \{e_X, e_Y, \ldots\}$ of nodes called observed nodes (these are associated with classical random variables, RVs) and the remaining nodes being unobserved (can be described by classical, quantum or post-quantum systems), and (2) a probability distribution $P_\cG (S)$ on the observed nodes of $\mathcal{G}$ which satisfies a linking property relative to $\mathcal{G}$.
Specifically, the affects framework employs $d$-separation \cite{Pearl1990,Pearl2009,arxiv.1302.3595} for this linking property,  ensuring that nodes $d$-separated in the graph are conditionally independent of one another.
This property holds in all acyclic (classical and non-classical) causal models \cite{Pearl2009, Henson2014, arxiv.1906.10726} as well as some cyclic causal models \cite{Bongers_2021, VVC}.
We note that the definition of causal models used here is minimal and top-down:
it does not refer to or make assumptions about the \emph{causal mechanisms} (systems, states, transformations or channels) of the underlying physical theory which is suitable for general no-go theorems.
This is in contrast to typical approaches to classical and non-classical causal models, that are bottom-up, that specify and study such causal mechanisms on a graph (e.g., \cite{Pearl2009,arxiv.1906.10726, Henson2014,Costa2016,Bongers_2021,Barrett2021, Ferradini2025quantum}). 

The framework then defines signalling between observed RVs of such a general causal model, via the notion of an \emph{affects relation}.
If intervening freely on a set of RVs $X \subset S$ changes the distribution of another set $Y \subset S$, then we say $X$ affects $Y$ and write $X \affects Y$. In particular, $X \affects Y$ allows one to infer that there exist $e_X \in X$ and $e_Y \in Y$ such that there is a directed path from $e_X$ to $e_Y$ in the causal structure $\mathcal{G}$ (or $e_X$ is a \emph{cause} of $e_Y$), which we will write succinctly as $X \cause Y$. 
In contrast, in so-called \emph{fine-tuned} causal models, where d-connection in $\cG$ does not necessarily imply a correspondent observable correlation, a causal relation $X \cause Y$ does not imply signalling from $X$ to $Y$, i.e.\ $X \affects Y$.

The framework also defines more general forms of signalling, via \emph{higher-order affects relations}. These take the form $X\affects Y \given \doo(Z)$ ($X$ affects $Y$ given $\doo(Z)$) and capture that a set of RVs $X$ signals to the set $Y$, given that some interventions have been performed on another set $Z$.
When $Z=\emptyset$, we call it a $0^\text{th}$-order affects relation that is equivalent to $X\affects Y$.
A formal definition of affects relations, encompassing both the 0$^\text{th}$ and the higher-order case, is given in \cref{def:affects-cond}.

Generally, a (higher-order) affects relation can contain redundancies, e.g., $X\affects Y \given \doo(Z)$ with $X=\{X_1,X_2\}$, where $X_2$ is a redundant RV without any causal connections.
To rule out such scenarios, a notion of irreducibility has been introduced for the first~\cite{VVC} (Irred$_1$) and the third~\cite{Grothus2024} (Irred$_3$) argument of an affects relation (i.e., the arguments $X$ and $Z$ for the relation $X \affects Y \given \doo(Z)$), respectively. 
If $X\affects Y \given \doo(Z)$ is Irred$_1$, then it allows us to infer that in any underlying causal model leading to this affects relation, for all $e_X\in X$ (as opposed to \emph{at least one} $e_X$ for the case without irreducibility), there exists $e_Y\in Y$ such that $e_X\cause e_Y$.
We write: for all $e_X\in X$, $e_X\cause Y$.
In \cite{Grothus2024}, we showed that if $X\affects Y \given \doo(Z)$ is Irred$_3$, one can infer that for all $e_Z\in Z$, $e_Z \cause Y$. %

\begin{figure*}[t]
	\centering
    \begin{subfigure}[b]{0.45\textwidth}
        \centering
        \begin{tikzpicture}[scale=.9]
            \node (X) at (0,0) {$X$};
            \node (Y) at (0,2) {$Y$};
            \node (A) at (2,0) {$A$};
            \node (B) at (2,2) {$B$};
            \node (Z) at (4,2) {$Z$};

            \begin{scope}[>={to[black]},
                          every edge/.style=affectsarrow]
                \path [->>] (A) edge (X);
                \path [->>] (X) edge (Y);
                \path [dashed,->>] (Y) edge (A);
                \path [dashed,->>] (Y) edge (B);
                \path [->>] (B) edge (Z);
                \path [dashed,->>] (Z) edge (A);
                \path [dashed,->>] (Z) edge[bend left] (B);
            \end{scope}
        \end{tikzpicture}
        \caption{Implications regarding the causal structure}
        \label{fig:ACL7-clg}
    \end{subfigure}%
    \begin{subfigure}[b]{0.45\textwidth}
        \centering
        \begin{tikzpicture}[dot/.style={circle,inner sep=1pt,fill,name=#1},scale=0.7]

            \node [dot=A,label=$\AA$] at (0,0) {};
            \node [dot=X,label=$\XX$] at (1,1) {};
            \node [dot=Y,label=$\YY$] at (2,2) {};
            \node [dot=Z,label=$\ZZ$] at (3.7,1.8) {};
            \node [dot=B,label=$\BB$] at (5,0.5) {};

            \node (left) at (-0.5,-0.5) {};
            \node (right) at (6,-0.5) {};
            \draw (left) -- ++(4,4);
            \draw (right) -- ++(-4,4);
            \fill[fill=blue!20] (2.75,2.75) -- (3.5,3.5) -- (2,3.5) -- (2.75,2.75);

            \node [dot=join,label=$a$] at (2.75,2.75) {};

        \end{tikzpicture}
        \caption{Embedding in ($1+1$)-Minkowski spacetime}
        \label{fig:ACL7-compat}
    \end{subfigure}%
	\caption{
        Let $X, Y, Z, A, B \in S$. Consider $X \affects Y, Y \affects AB, A \affects X, Z \affects AB, B \affects Z$.
        \textbf{(a)} No matter the underlying causal model, the affects relations allow to infer that a causal loop will always be present.
        Here, $\cause$ denotes causal relations between two RVs, while its dashed version stands for alternative causal relations, e.g.\ $Y \cause AB$ being equivalent to $(Y \cause A \ \text{or} \ Y \cause B)$, so at least one of the dashed arrows from $Y$ must be present as a causal relation. 
        \textbf{(b)} A compatible, non-degenerate embedding of these affects relations into ($1+1$)-Minkowski spacetime, where all ORVs need be located on certain light-like surfaces.
    }
	\label{fig:ACL7}
\end{figure*}
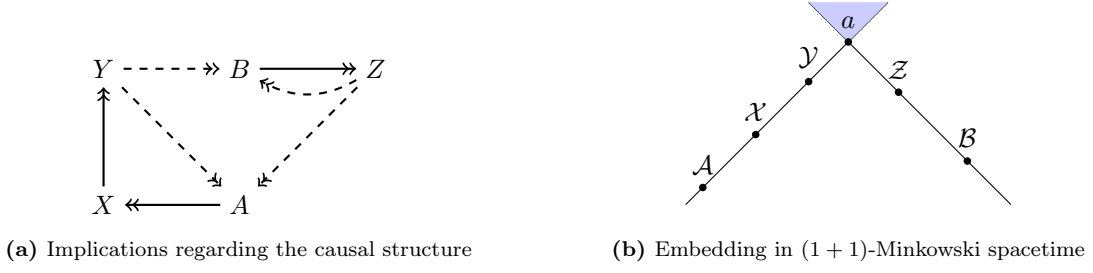%

{\bf Spacetime embedding, NSS and causal loops} The affects framework then also defines an \emph{embedding} of the set of RVs $S$ of a causal model into a spacetime, modelled as a partially ordered set (\emph{poset}) $\TT$ where the order captures the spacetime's light cone structure. This is done by assigning each RV $e_X$ a \emph{location} $O(e_X) \in \TT$, defining an ordered RV (ORV) $e_\XX := (e_X, O(e_X))$.
With the future light cone of an ORV being defined as
$\Fut (e_\XX) := \{ a \in \TT : a \succeq O(e_X) \}$, the joint future of a set $\XX$ of ORVs is the intersection of their respective future light cones, $\Fut_s (\XX) \coloneqq \bigcap_{e_\XX\in \XX} \Fut (e_\XX)$.
Having embedded information-theoretic causal models in spacetime, the relativistic principle of no superluminal signalling (NSS) is formalised as a compatibility condition between the affects relations of the causal model (capturing signalling) and the partial order of $\TT$ (capturing light cone structure).
Specifically, for each affects relation $X \affects Y \given \doo(Z)$ we require
\[
    X \affects Y \given \doo(Z) \ \text{is} \ \mathrm{Irred}_1
    \
    \implies
    \
    \Fut_s (\YY) \cap \Fut_s(\ZZ) \subseteq \Fut_s (\XX) \, .
\]
In \cite{VVC} it is further explained how this condition captures the NSS principle in spacetime, and Irred$_1$ plays a key role here. Further, an embedding is considered \emph{non-degenerate} if no two RVs are embedded into the same location.

Based on the causal inference rules mentioned above, we may (potentially) infer operationally detectable information-theoretic cycles from affects relations (i.e.\ signalling) alone~-- so-called \emph{affects causal loops} (ACLs)~\cite{VVC}.
The simple-most example of such a loop is given by two RVs affecting one another, i.e.\ both $e_X \affects e_Y$ and $e_Y \affects e_X$, where NSS implies a degenerate embedding with $O(e_X) = O(e_Y)$. Several other complex classes of ACLs were defined in \cite{VVC}.
In \cref{fig:ACL7} we illustrate a more complex example of such a causal loop which admits a non-degenerate and compatible (i.e., satisfying NSS) embedding in $(1+1)$-Minkowski spacetime. 

\begin{figure*}
    \centering
    \begin{subfigure}[b]{0.45\textwidth}
        \centering
		\begin{tikzpicture}[dot/.style={circle,inner sep=1pt,fill,name=#1},scale=1]

			\node [dot=A,label=$a$] at (1,1) {};
			\node [dot=X,label=$x$] at (2,2) {};
			\node [dot=Z,label=$y$] at (3.7,1.8) {};
			\node [dot=B,label=$b$] at (4.5,1) {};

			\node (left) at (0.5,0.5) {};
			\node (right) at (5,0.5) {};
			\draw (left) -- ++(3.5,3.5);
			\draw (right) -- ++(-3.5,3.5);
			\fill[fill=blue!20] (2.75,2.75) -- (4,4) -- (1.5,4) -- (2.75,2.75);
		\end{tikzpicture}
        \caption{$d=1$}
        \label{fig:d-1}
    \end{subfigure}%
    \begin{subfigure}[b]{0.55\textwidth}
        \centering
        \begin{tikzpicture}[scale=0.82]
            \begin{scope}
                \clip (0.0,0) circle (2.5);
                \fill[fill=red!20] (3.0,0) circle (2.5);
            \end{scope}

            \begin{scope}
                \clip (0.75,0) circle (1.75);
                \fill[fill=blue!20] (2.25,0) circle (1.75);
            \end{scope}

            \draw[fill=none](0.75,0) circle (1.75);
            \draw[fill=none](0.00,0) circle (2.50);
            \draw[fill=none](2.25,0) circle (1.75);
            \draw[fill=none](3.00,0) circle (2.50);

            \draw[fill=black](0.75,0) circle (1pt) node [above] {\small $x$};
            \draw[fill=black](0.00,0) circle (1pt) node [above] {\small $a$};
            \draw[fill=black](2.25,0) circle (1pt) node [above] {\small $y$};
            \draw[fill=black](3.00,0) circle (1pt) node [above] {\small $b$};
        \end{tikzpicture}
        \caption{$d=2$}
        \label{fig:d-2}
	\end{subfigure}%
    \caption{Representation of light cones in Minkowski spacetime for $d$ spatial dimensions. \textbf{(a)} For $d=1$ Minkowski spacetime does not show conicality. This is because $L_1 \coloneqq \{a,b\}$ and $L_2 \coloneqq \{x,y\}$ are distinct, yet share the same joint future. \textbf{(b)} For $d=2$, Minkowski spacetime is conical. As can be seen from the figure which shows one particular time slice, the joint futures are distinct between $L_1$ and $L_2$ (and the red region is thus not covered by the blue one).}
	\label{fig:d}
\end{figure*}
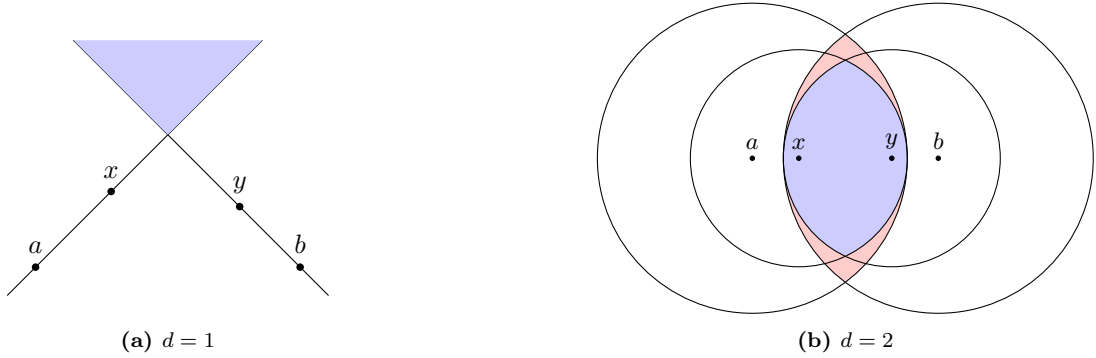

{\bf Conicality~-- an order-theoretic property of spacetime geometry} 
When embedding the RVs of a causal model into spacetime, modelled as a poset $\TT$, the \enquote{shape} of $\TT$ may constrain the order-theoretic properties of the embedding.
One such property is \emph{conicality} \cite{Grothus2024}, capturing the requirement that the joint future region $\bigcap_{x\in L} \bar J(x)$ (intersection of future light cones $\bar J(x) := \{ a \in \TT : a \succeq x \}$) of a finite set $L \subset \TT$ uniquely determines the location of all points in $L$ that contribute non-trivially to $\bigcap_{x\in L} \bar J(x)$.
In particular, this property is satisfied by ($d+1$)-Minkowski spacetime for $d>1$ spatial dimensions, while it is violated for $d=1$, as shown in \cref{fig:d}.
Moreover,~\cite{Paganini2026} establishes that conicality holds for any causally simple, \emph{future cohesive} spacetime of dimension $d+1$ with $d \ge 2$.

\section{Main Theorem}
\label{sec: main}
\noindent \textbf{Theorem 1.}
In conical spacetimes, all compatible embeddings of ACLs are degenerate.
\medskip

Since compatibility precisely captures the NSS principle \cite{VVC,VVCPR}, this tells us that when non-degenerately embedding a causal model in any conical spacetime, the NSS principle does rule out all operationally detectable causal loops.
In particular, the ACL of \cite{VVC} and of \cref{fig:ACL7} are compatibly and non-degenerately embedded in non-conical ($1+1$)-Minkowski spacetime. 
The theorem, formally presented in \cref{sec:proof}, is proven in two broad steps, that also lead to useful technical results.

{\bf Step 1: Reducing higher-order to equivalent $0^\text{th}$-order affects relations} (\cref{thm:deorder})
We show that for any set of possibly higher-order affects relations, there exists a set of $0^\text{th}$-order affects relations, with the same implications for causal inference and compatibility conditions in conical spacetimes.
Specifically, whenever there is a higher-order affects relation $X \affects Y \given \doo(Z)$ satisfying Irred$_1$, there exists some $s_Z \subseteq Z$, potentially empty, such that $X \affects Y \given \doo(s_Z)$ is Irred$_3$.
Then, the combination of both relations has the same implications for causal inference and imposes the same compatibility constraints (in conical spacetimes) as a $0^\text{th}$-order relation $Xs_Z \affects Y$ being Irred$_1$.
In proving this result, we also advance the understanding of the interplay between Irred$_1$ and Irred$_3$, which is not yet well-understood. It follows that the original set of (possibly higher-order) affects relations implies, via causal inference, an ACL if and only if the attained set of $0^\text{th}$-order affects relations also implies an ACL. In conical spacetimes, the original affects relations similarly satisfy compatibility (NSS) for an embedding if and only if the reduced $0^\text{th}$-order relations do.
Therefore, we now restrict, w.l.o.g., our attention to ACLs formed by $0^\text{th}$-order affects relations.

{\bf Step 2: Ruling out ACLs in $0^\text{th}$-order affects relations embedded in spacetime} (\cref{thm:conical-embedding-no-loops}) We sketch the proof for the case where the ACL is formed only by $0^\text{th}$-order affects relations, illustrating the main technique intuitively by means of an example.
Generally, the causal inference possible from an (0$^\text{th}$-order) affects relation $X \affects Y$ being Irred$_1$ is limited, allowing one to infer only that each $e_X \in X$ is a cause of \emph{some} element $e_Y$ of the second argument $Y$.
Therefore, a causal loop is implied by a set of such affects relations
if and only if no matter which element in the second argument of each affects relation in such a set we \enquote{choose} to be causally influenced by the first argument, the resulting causal structure will be cyclic.
Considering \cref{fig:ACL7-clg} as an example, we have some causal loop no matter whether $Y$ causes $A$ or $B$ (implied by $Y \affects AB$), and whether $Z$ causes $A$ or $B$ (implied by $Z \affects AB$).
We now evaluate compatibility \enquote{along} the causal arrows (which align with the spacetime causal structure), starting from ORV $\mathcal{A}$. 
For a compatible and non-degenerate embedding of these affects relations in any spacetime, $A \affects X \affects Y$ gives $ \Fut(\AA) \supseteq \Fut (\YY)$ and $Y\affects AB$ gives $\Fut (\YY) \supseteq \Fut_s(\AA\BB)$. Next, we show that in conical spacetimes, compatibility further implies strictness of these subset relations on the futures, explicitly: $
    \Fut(\AA) \supsetneq \Fut (\YY) \supsetneq \Fut_s(\AA\BB).
$ Continuing onward for $\BB$ separately, through the same compatibility + conicality implications, we obtain
$
    \Fut(\BB) \supsetneq \Fut(\ZZ) \supsetneq \Fut_s(\AA\BB).
$
Moreover, conicality allows to combine these statements while preserving the strict subset relations, i.e., $\Fut_s(\AA\BB) \supsetneq \Fut_s(\YY\ZZ)$ and $\Fut_s(\YY\ZZ) \supsetneq \Fut_s(\AA\BB)$. (cf.\ \cref{thm:span-join}, note that neither of these is satisfied in the non-conical spacetime embedding shown in \cref{fig:ACL7-compat}). We therefore receive a contradiction of the form
$\Fut_s(\AA\BB) \neq \Fut_s(\AA\BB)$.

{\bf A generalisation and notion of fine-tuning for spacetime embeddings} We also provide a generalisation of the main theorem by defining the notion of 
 \emph{conical embeddings}, which consider the condition of conicality only for the futures of all points $O(\SC)$ in the image of the embedding of the RVs $S$ and the intersections of their futures. While any embedding into a conical poset $\TT$ is trivially conical, we can have non-conical $\TT$ within which there exists a conical embedding.
 This shows that causal loops without superluminal signalling require a form of fine-tuning in the spacetime embedding, captured by non-conicality of the embedding~-- in addition to fine-tuning in the causal model, as noted in \cite{VVC}. Indeed, non-conical embeddings require the joint futures of distinct sets of spacetime events to align in a very precise ``fine-tuned'' manner, as we also see in \cref{fig:ACL7} and \cref{fig:d-1}.

{\bf Causal loops require a specific form of information-theoretic fine-tuning}
Moreover, we also strengthen the observation of \cite{VVC}, by showing that a specific form of fine-tuning in causal models is required for causal loops to be embeddable without superluminal signalling.
Specifically, this is signified through \emph{clustering} of affects relations \cite{Grothus2024}, where certain signalling can only be detected in the correlations between multiple RVs.
Further details and a proof of this result are presented in \cref{sec:clus}.

\section{Discussion}
Our results show that in $(d+1)$-Minkowski spacetime for any $d>1$---and more generally, in any conical spacetime---no-superluminal signalling does rule out operationally verifiable causal loops, thereby fully formalising and proving an informal conjecture made in~\cite{VVC_Letter}, where this was shown to fail in $(1+1)$-Minkowski. The result holds without assumptions on the theory---classical, quantum, or post-quantum---governing the underlying cyclic causal models.
We further show that any causal loop compatible with NSS in a spacetime requires \emph{both} fine-tuning in the causal model \emph{and} in the spacetime embedding: an absence of (superluminal) signalling can turn into (superluminal) signalling under slight perturbations of either the probability distribution or the spacetime locations. This extends the correspondence between information-theoretic and spatio-temporal causal structures beyond the connections established in~\cite{Grothus2024}.

There are two immediate directions for further extending our no-go result. First, when the underlying causal model and system cardinalities are unknown, the causal inference rules used here can be regarded as ``complete'' in the sense of~\cite{Master}; but it is not known whether this remains true once such properties are known or fixed. Could stronger inference rules from affects relations be derived when restricting, e.g., to binary random variables, such that additional sets of affects relations would certify causal loops? Second, the affects framework constrains cyclic causal structures via the $d$-separation property. While $d$-separation holds in many cyclic scenarios, there do exist cyclic models that violate it~\cite{Neal2000}, motivating alternative graph-separation properties such as $\sigma$-separation~\cite{arxiv.1710.08775} and $p$-separation~\cite{Ferradini2025classical}, the latter applying to classical and quantum cyclic models. Whether our no-go result extends to cyclic causal models satisfying $\sigma$- or $p$-separation---and whether spacetime-embeddable causal loops could arise from cyclic models linked to valid process matrices ~\cite{Barrett2021,Ferradini2025quantum}---are natural future directions.

More broadly, although our analysis fixed a background spacetime, the underlying frameworks are more versatile: cyclic information-theoretic causal models are formulated independently of spacetime, and spacetime itself is modelled as a partially ordered set (poset), without assuming manifold structure. Since posets can equivalently be represented as directed acyclic graphs (DAGs), and both of these single out a global direction,  our results can thus be viewed as identifying properties of DAGs compatible with a given possibly cyclic information-theoretic causal structure, with no superluminal signaling and no superluminal causation imposing distinct compatibility conditions. This would allow to study how acyclic causal order and a global direction of time may emerge from a general operational causal structure. This makes the approach relevant for studying realisable protocols in a background spacetime, and also the \emph{emergence} of acyclic causal order and time direction without presupposing a spacetime. The latter was recently explored in a related framework for cyclic quantum causal modelling~\cite{Ferradini2025quantum} (where $p$-separation was introduced) and its bridge to (holographic) tensor networks~\cite{ferradini2026emergent}, where fine-tuning and relativistic principles also play a key role in identifying an emergent time direction from operational primitives.

\section*{Acknowledgements}
MG acknowledges financial support by l’Agence Nationale de la Recherche (ANR), project ANR-22-CE47-0012.
VV acknowledges support from a government grant managed by the Agence Nationale de la Recherche under the Plan France 2030 with the reference ANR-22-PETQ-0007.
For the purpose of open access, the authors have applied a CC-BY public copyright licence to any Author Accepted Manuscript (AAM) version arising from this submission.

\bibliography{references_thesis.bib}

\newpage
\appendix
\begin{center}
 {\bf APPENDIX}   
\end{center}

\section{Proof of main theorem}
\label{sec:proof}

To state the main theorem in its most general form, we formally give the most general form of an affects relation, which allows conditioning on an additional RV $W$. Here, $\cG_{\doo(X)}$ is obtained from $\cG$ by removing all the incoming edges of $X$, while $P_{\cG_{\doo(X)}}$ denotes a corresponding post-intervention distribution.
Together, $\cG_{\doo(X)}$ and $P_{\cG_{\doo(X)}}$ (satisfying the same linking property of $d$-separation) define a post-intervention causal model for interventions on $X$.
We note that in bottom-up causal modelling approaches that specify causal mechanisms (such as states and channels of a theory) on a graph, a do-intervention on any $X$ in a given pre-intervention causal model 
fully specifies the post-intervention model and its distribution $P_{\cG_{\doo(X)}}$: the post-intervention model is obtained by replacing the causal mechanism for $X$ with $X = x$ while keeping other mechanisms unchanged i.e., forcing $X$ to take value $x$ independently of its parents.
Further details on do-interventions in the affects framework can be found in \cite{VVCJ}. 

\begin{definition}[Conditional (higher-order) affects relations~\cite{VVC}]
    \label{def:affects-cond}
    Consider a causal model over a set $S$ of observed nodes, associated with a causal structure $\mathcal{G}$.
    For pairwise disjoint subsets $X, Y, Z, W \subset S$, with $X, Y$ non-empty, we say
    $
        X \, \text{affects} \; Y \, \text{given} \, \doo(Z), W \, ,
    $
    which we concisely denote as
    $
        X \vDash Y \,|\, \doo(Z), W \, ,
    $
    if there exist values $x$ of $X$, $z$ of $Z$ and $w$ of $W$  such that
    \begin{equation}
        \begin{split}
        &\;P_{\mathcal{G}_{\doo(XZ)}} (Y | X=x, Z=z, W=w) \\ \neq
        &\;P_{\mathcal{G}_{\doo(Z)}} (Y | Z=z, W=w)\,.
        \end{split}
    \end{equation}
    For $W \neq \emptyset$, we speak of a \emph{conditional affects relation}, denoted by $X \affects Y \given W$.
    If $Z \neq \emptyset$, we have a \emph{higher-order (HO) affects relation}, denoted by $X \affects Y \given \doo(Z)$.
    The trivial case of $W$ or $Z = \emptyset$, we have an \emph{unconditional} or  \emph{0$^\text{th}$-order} affects relation, respectively.
\end{definition}

Further, we generalise from conical spacetimes to conical embeddings in arbitrary spacetimes, which are defined as follows:

\begin{definition}[Conical Embedding~\cite{Grothus2024}]
    \label{def:conical-embedding}
    An embedding of a set $\mathcal{S}$ of ORVs in a spacetime $\TT$ is called \emph{conical} if for any two finite subsets $L_i, L_j \subseteq O(\SC) = \{O(X) | X \in S \} \subseteq \TT$,
    $
        f(L_i) = f(L_j)
        \ \implies \ 
        \spann(L_i) = \spann(L_j) 
    $
    holds.
\end{definition}

Intuitively, with this definition, we consider the condition of conicality only for the futures of all points $O(\SC)$ in the image of the embedding and the intersections of these futures.
Hence, any embedding into a conical poset $\TT$ is trivially conical. However, the conicality of $\TT$ is not necessary to obtain a conical embedding:
Actually, the vast majority of embeddings into ($1+1$)-dim.\ Minkowski spacetime are conical as well.

\begin{remark}
    \label{rem:fine-tuned-embedding}
    Specifically, such embeddings can be non-conical only if one ORV is embedded into the future light cone surface of another,
    constituting a kind of fine-tuning of the embedding.
    As a possible sufficient criterion for being located \enquote{on the light cone surface} we suggest, in purely order-theoretical terms:
    For $a, b \in \TT$, $b$ is embedded into the future light cone surface of $a$ iff $M = \bar{J}^+ (a) \cap \bar{J}^- (b) = \{ x \in \TT | a \preceq x \preceq b \}$ is totally ordered by $\prec$, i.e.\ $\not\exists \, b, c \in M : \ b \unord c$.
    Here, $M$ is a degenerate example of a causal diamond~\cite{Hounnonkpe2019, Witten2020}.
    For Minkowski spacetime, this criterion is also necessary for being located on the light cone surface, while singularities or other causal irregularities could lead to situations where this condition is not necessary for being on the light cone surface.
\end{remark}

Having this concept established, we return to the proof of the theorem by establishing three auxiliary lemmas, the first of which precisely captures how conicality impacts the structure of the embedding to enable our proof. 

\begin{lemma}
    \label{thm:span-join}
    Let $\TT$ be a poset and $\SC$ a set of ORVs embedded non-degenerately and conically therein.
    Let $\YY_i \subset \SC$ and $\XX_i \in \SC \setminus \YY_i$ for all $i$, and $\XX = \spann(\XX) = \{ \XX_i \}_i, \YY = \bigcup_i \YY_i$.
    Then,
    $
        \Fut_s (\YY_i) \subsetneq \Fut (\XX_i) \ \forall i
        \ \implies \
        \Fut_s (\YY) \subsetneq \Fut_s (\XX) \, .
    $
\end{lemma}
\begin{proof}
    By conjunction, we arrive at
    $\Fut_s (\YY_i) \subsetneq \Fut (\XX_i) \,\forall i$\\
    $
        \, \implies \,
        \Fut_s (\YY) \subsetneq \Fut (\XX_i) \, \forall i
        \, \implies \,
        \Fut_s (\YY) \subseteq \Fut_s (\XX) .
    $

    Due to conicality of the embedding, $\Fut_s (\YY) \overset{!}{=} \Fut_s (\XX)$ would imply that the set of locations of $\spann(\YY)$ is equal to the set of locations of $\spann(\XX)$.
    As the embedding is non-degenerate, this would imply the RVs themselves to be the same, i.e.\ $\spann(\XX) = \spann(\YY)$, and moreover, $\XX = \spann(\YY)$.
    (Note that $\XX_i$ and $\YY$ are not necessarily disjoint.)
    However, as by assumption $\XX_i \in \SC \setminus \YY_i$, we find that $O(\XX_i) \not\in O(\YY_i)$ due to non-degeneracy.
    Due to conicality, $\Fut_s (\YY_i) \subsetneq \Fut (\XX_i)$ then also affirms $\XX_i \not\in \spann(\XX_i \YY_i) = \spann(\YY_i)$. 
    This implies $\XX_i \not\in \spann(\YY)$, which poses a contradiction. $\lightning$
\end{proof}

\begin{lemma}
    \label{thm:deduce-affects}
    For any causal model with observed RVs $S$, $X, Y, Z \subset S$ disjoint and $e_X \in X, e_Z \in Z$,\\
    $e_X \naffects Y \given \doo(Z X \setminus e_Z e_X)
    \ \land \
    e_X \affects Y \given \doo(ZX \setminus e_X)$\\
    implies\\
    $e_Z \affects Y \given \doo(Z X \setminus e_Z)
    \ \lor \
    e_Z \affects Y \given \doo(Z X \setminus e_Z e_X)$.
\end{lemma}
\begin{proof}
    By definition, the relations expand to
    \begin{align*}
        P_{\cG_{\doo(ZX \backslash e_Z e_X)}} (Y | ZX \setminus e_Z)
        &=
        P_{\cG_{\doo(ZX \backslash e_Z)}} (Y | ZX \backslash e_Z e_X) \\
        P_{\cG_{\doo(ZX \backslash e_X)}} (Y | ZX)
        &\neq
        P_{\cG_{\doo(ZX)}} (Y | ZX \backslash e_X)\,.
    \end{align*}
    For these equations to be compatible with one another, it is required that either $P_{\cG_{\doo(ZX \backslash e_Z e_X)}} (Y | ZX \setminus e_Z) \neq P_{\cG_{\doo(ZX \backslash e_X)}} (Y | ZX)$ or $P_{\cG_{\doo(ZX \backslash e_Z)}} (Y | ZX \setminus e_Z e_X) \neq P_{\cG_{\doo(ZX)}} (Y | ZX \setminus e_X)$, exactly corresponding to $e_Z \affects Y \given \doo(Z X \setminus e_Z)$ or $e_Z \affects Y \given \doo(Z X \setminus e_Z e_X)$.
\end{proof}
This proof generalises to subsets $s_Z \subset Z$ and $s_X \subset X$. However, this generalisation will not be required in this work.
With this at hand, we continue by stating the formal version of \textbf{Step 1} (described in \cref{sec: main}), slightly generalised to cover \emph{conditional} higher-order affects relations $X \affects Y \given \doo(Z), W$ (cf.\ \cref{def:affects-cond}).

\begin{lemma}
    \label{thm:location-symmetry}
    Consider a non-degenerate conical embedding into a poset $\TT$.
    Let $\YY \subset \SC$, $\XX = \{ e_\XX^i \}_i \subset \SC$ with
    $\Fut(e_\XX^i) \supseteq \Fut_s(\YY \XX \setminus e_\XX^i) \ \forall\, i$.
    Then $\Fut(e_\XX^i) \supseteq \Fut_s(\YY) \ \forall\, i\,$.
\end{lemma}
\begin{proof}
    This follows by applying an argument perfectly analogous to the proof of Lemma~G.9 of~\cite{Grothus2024}:
    After using Lemma~G.7 of that paper (with  $\XX$ and $\YY$ taking the role of $\mathcal{B}$ and $\mathcal{A}$ there), we obtain
    $\Fut(\YY \XX \setminus e_\XX^i) = \Fut_s(\YY\XX \setminus e^j_\XX) \ \forall\, i,j$ and apply the implication given in Definition~G.4 
    to obtain the statement of the lemma.
\end{proof}    

We now continue by stating the main lemma corresponding to step 1 of the procedure outlined in the main text.

\begin{lemma}
    \label{thm:deorder}
    Consider a causal model featuring a (conditional) higher-order affects relation $X \affects Y \given \doo(Z), W$ satisfying Irred$_1$, which w.l.o.g.\ reduces to
    $X \affects Y \given \doo(s_Z), W$ being Irred$_3$ for some $s_Z \subseteq Z$, potentially empty.
    Then
    \begin{equation}
        \label{eq:deorder-caus}
        \forall e_{XZ} \in X s_Z : \ e_{XZ} \cause YW \,.
    \end{equation}
    
    For a conical embedding, we additionally find
    \begin{equation}
        \label{eq:deorder-comp}
        \Fut_s(\XX) \cap \Fut_s(s_\ZZ) \supseteq \Fut_s (\YY) \cap \Fut_s(\WW)\,.
    \end{equation}
    Altogether, the causal inference and compatibility conditions for the combination of $X \affects Y \given \doo(Z), W$ Irred$_1$ and $X \affects Y \given \doo(s_Z), W$ Irred$_3$ are the same as those of $X s_Z \affects YW$ with Irred$_1$.
\end{lemma}
\begin{proof}
    By Lemma IV.8 of~\cite{VVC}, $X \affects Y \given \doo(Z), W$  being Irred$_1$ implies $X \affects YW \given \doo(Z)$ being Irred$_1$ for each causal model.
    Then, $X \affects YW \given \doo(Z)$ can either be Red$_3$ or not.
    In any case there exists $X \affects YW \given \doo(s_Z)$ for some $s_Z \subseteq Z$, satisfying either Irred$_3$ or having $s_Z = \emptyset$.

    However, it is not necessarily clear that this relation is again irreducible in the first argument, or at least, if the relevant implications for causal inference and compatibility are maintained.
    
    By the original relation being Irred$_1$, we have $e_X \affects YW \given \doo(ZX \setminus e_X) \ \forall e_X$.
    We differentiate two cases:

    \textbf{I)} If $\forall e_Z \in Z \,\ \exists e_X \in X: e_X \naffects YW \given \doo(Z X \setminus e_Z e_X)$, we can use this fact together with Irred$_1$ of the original relation implying $e_X \affects YW \given \doo(ZX \setminus e_X)$ to deduce additional affects relations.
    Following \cref{thm:deduce-affects}, they jointly certify
    $e_Z \affects YW \given \doo(Z X \setminus e_Z) \ \lor \ e_Z \affects YW \given \doo(Z X \setminus e_Z e_X) \quad \forall e_Z \in Z$.\footnote{
        This property is analogous to imposing Irred$_3$ for $X \affects YW \given \doo(Z)$ only with respect to $e_Z \in Z$ rather than all $s_Z \subseteq Z$, leading to a different notion for $\abs{Z} > 2$.
    }
    For causal inference, either of these alternative affects relations implies $e_Z \cause YW \ \forall e_Z \in Z$, affirming \cref{eq:deorder-caus}.
    
    By compatibility, they either at least imply
    $\Fut(e_\ZZ) \supseteq \Fut_s(\ZZ \setminus e_\ZZ) \cap \Fut_s (\XX) \cap \Fut_s (\YY\WW)$,
    which together with $\Fut(e_\XX) \supseteq \Fut_s(\XX \setminus e_\XX) \cap \Fut_s(\ZZ) \cap \Fut_s(\YY\WW)$ (from $e_X \affects YW \given \doo(ZX \setminus e_X)$) can be equivalently written as $\Fut(e_{\ZZ\XX}) \supseteq \Fut_s(\YY\WW\XX\ZZ \setminus e_{\XX\ZZ})$ for all $e_{\ZZ\XX}\in \ZZ\XX$.
    By \cref{thm:location-symmetry}, this implies $\Fut(e_{\ZZ\XX}) \supseteq \Fut_s (\YY\WW) \ \forall e_{\ZZ\XX} \in \ZZ\XX$, recovering \cref{eq:deorder-comp}.
    
    We will repeat the argument for case $\mathbf{I}$ for subsets $s_Z$ in place of $Z$ in later iterations, see below in \textbf{II}.

    \textbf{II)} Otherwise, $\exists e_Z \in Z \;\ \forall e_X \in X: e_X \affects YW \given \doo( Z X \setminus e_Z e_X)$\footnote{
        For $\abs{X} > 2$, this is not quite sufficient for irreducibility of $X\affects YW \given \doo(Z \setminus e_Z)$, which is precisely the analogous statement obtained when replacing $e_X$ with $\forall s_X \subset X$.
        In isolation, it implies $\Fut(e_\XX) \supseteq \Fut_s(\XX \setminus e_\XX) \cap \Fut_s(\ZZ \setminus e_\ZZ) \cap \Fut_s (\YY\WW) \ \forall e_\XX \in \XX$.}.
    By \cite[Lemma~F.2]{Grothus2024}, for a conical embedding, this has the same implications for compatibility (and additionally, for causal inference) as
    $X \affects YW \given \doo(Z \setminus e_Z)$ being Irred$_1$. If $Z \setminus e_Z = \emptyset$, we arrive at the claim for $s_Z = \emptyset$.
    If $Z \setminus e_Z \neq \emptyset$, we repeat the same argument for the relations $e_X \affects YW \given \doo(ZX \setminus e_Z e_X) \ \forall e_X \in X$ in place of $e_X \affects YW \given \doo(ZX \setminus e_X) \ \forall e_X \in X$, again considering for this set of relations whether case $\textbf{I}$ or $\textbf{II}$ applies with $Z \setminus e_Z$ in place of $Z$.
    Repeating recursively, we ultimately either end up in case $\textbf{I}$ to recover the claim or, after $\abs{Z}$ repetitions, reduce to $s_Z = \emptyset$ and recover it for that case.
\end{proof}
Essentially, while this proof does not affirm $X \affects YW \given \doo(s_Z)$ to be Irred$_1$, it recovers the associated implications of irreducibility for causal inference and compatibility.
However, as raised in the footnotes, case \textbf{I} in the proof corresponds to the case where the relation is irreducible in the third argument in a weaker sense only, while case \textbf{II} represents the case where the relation is reducible in the first argument in this weaker sense. We discuss these properties in \cref{sec:irreducible-weak}.

This result generalises from a single affects relations to entire set of affects relations, which may ultimately give rise to an affects causal loop.

\begin{corollary}
    \label{thm:deorder-loop}
    Consider a set of higher-order affects relations $\mathscr{A}$.
    Then, \cref{thm:deorder} yields a set of 0$^\text{th}$-order affects relations $\mathscr{A}'$ which imply an ACL if and only if $\mathscr{A}$ implies an ACL.
    Furthermore, under the conditions of that lemma, both $\mathscr{A}$ and $\mathscr{A}'$ yield the same constraints for compatible conical embeddings into a spacetime $\TT$.
\end{corollary}
\begin{proof}
    To demonstrate this, one needs to rule out the existence of further causal inference rules to the ones reviewed above, which could contribute to an ACL.
    For an individual affects relation $X \affects Y$, it is consistent with a causal model where we have just a single causal arrow $e_X \cause e_Y$ for some $e_X \in X,\ e_Y \in Y$.
    Hence, when considering individual relations only, the causal inference rule reviewed in the main text is generally the best one can do.
    When applying this argument to an affects relation irreducible in the first or third argument, which ultimately represents a set of affects relations, an analogous argument applies to recover the causal inference rules reviewed above.
    
    The conjunction of multiple affects relations --~or beyond that, information on the cardinality of RVs, their probability distribution or other constraints on the causal structure~-- could however in some cases allow for additional causal inference.
    However,~\cite[Lemma 3.12]{Master} provides demonstrates that in absence of such restrictions, one can generally give a causal model featuring exactly the causal influences implied by the inference rules reviewed above.    \footnote{
    More specifically, RVs of sufficient cardinality may be fine-grained into independent (Cartesian) factors, with each affects relation only interacting with a subset of these factors.~\cite{VVR, Grothus2024}
}
\end{proof}

\begin{theorem}
	\label{thm:conical-embedding-no-loops}
    Any conical compatible embedding of ACLs into spacetime is degenerate.
\end{theorem}
\begin{proof}
    Consider a set of affects relations which imply an ACL.
    Then by \cref{thm:deorder-loop}, we can transform it into a set of unconditional 0$^\text{th}$-order affects relations, with the same implications for causal inference and compatibility, and continue with this transformed set.
    Therefore, in the rest of the proof, we only consider sets of 0$^\text{th}$-order affects relations entailing ACLs.

Choose an ORV $e_\XX$ that is part of a causal loop, as inferrable from the presence of a set of irreducible affects relations.
    Then, it must be part of at least one irreducible affects relation $X \affects Y$ implying $e_X \cause Y$ with $Y \subset S \setminus e_X$.
    This yields
    $\Fut_s (\spann(\YY)) \subsetneq \Fut (e_\XX)$ for any non-degenerate embedding due to conicality and $e_\XX \not\in \spann(\YY)$.

    Iterating, we consider $\spann(\YY)$. If $\YY$ is to be part of a causal loop, for all $e^i_\YY \in \spann(\YY)$, there is $Z_i \subset Z$ such that we can infer $e^i_Y \cause Z_i$ due to an irreducible affects relation.\footnote{
        If there were $e^i_\YY \in \spann(\YY)$ such that there is no such $Z_i$, $e_X \cause e^i_Y$ would be one option to resolve $e_X \cause Y$, with $e^i_Y$ being a childless node in the causal structure.
        Therefore, the causal relation $e_X \cause Y$ could not be part of a causal loop, would be in contradiction to our assumption.}
    This yields
    $
        \Fut_s (\spann(\ZZ_i)) \subsetneq \Fut (e^i_\YY) \: \forall i \ \implies \
        \Fut_s (\spann(\ZZ)) \subsetneq \Fut_s (\spann(\YY))
    $
    by \cref{thm:span-join}.

    Iterating by the same argument, we can chain these strict subset relations with one another.
    However, the graph is finite and due to $\spann(\AA) \neq \emptyset$ for any set of ORVs $\AA$, the chain continues indefinitely~-- as otherwise, we would not have a loop.
    Hence, there exist sets $\SC$ which appear multiple times in the chain, yielding $\Fut_s (\spann(\SC)) \neq \Fut_s (\spann(\SC))$.
    This poses a contradiction, thereby proving the claim. $\lightning$
\end{proof}

\section{No Causal Loops without Clustering} 
\label{sec:clus}

Consider an affects relation $X \affects Y \given \doo(Z)$, with $X,Y,Z \subset S$.
Then \emph{clustering} encodes the idea that the affects relation no longer holds when restricting to strict subsets of $X$, $Y$ or $Z$.
In the context of this work, we will primarily work with clustering in the interventional arguments, for which we give formal definitions.

\begin{definition}[Clustering in the first argument~\cite{Grothus2024}]
\label{def: clus1}
  An affects relation $X\vDash Y \given \mathrm{do}(Z)$ is called clustered in the first argument (denoted Clus$_1$) if $|X|\geq 2$ and there exists no $s_X\subsetneq X$ such that   $s_X\vDash Y \given \mathrm{do}(Z)$. 
\end{definition}

\begin{definition}[Clustering in the third argument~\cite{Grothus2024}]
\label{def: clus3}
  An affects relation $X\vDash Y \given \mathrm{do}(Z)$ is called clustered in the third argument (denoted Clus$_3$) if $|Z|\geq 1$ and there exists no $s_Z\subsetneq Z$ such that $X\vDash Y \given \mathrm{do}(s_Z)$. 
\end{definition}
Generally, clustering is a signature of a specific form of fine-tuning of the underlying causal mechanisms \cite{Grothus2024}.

As established there, the presence of clustering in particular allows to vastly simplify the characterisation of signalling and causal inference in a given causal model.

\begin{corollary}
    \label{thm:no-clus-conditional}
    Consider an affects relation $X \affects Y \given \doo(Z), W$ in a causal model without clustering (in the first three arguments).
    Then there exist $e_X \in X,\ e_{YW} \in YW$ such that $e_X \affects e_{YW}$.
\end{corollary}
\begin{proof}
    This follows as direct generalisation of Corollary~D.7 of \cite{Grothus2024} to conditional affects relations.
\end{proof}

\begin{lemma}
    An analogous statement to \cref{thm:deorder}
    holds when restricting to causal models which feature no clustering in the third argument rather than restricting to conical embeddings.
\end{lemma}
\begin{proof}
    Similarly to the proof of \cref{thm:deorder}, consider $X \affects Y \given \doo(Z), W$ which is Irred$_1$ and may or may not be Irred$_3$.
    For causal models without clustering
    in the third argument (cf.\ \cref{def: clus3}) of \emph{any} affects relations, it is easy to see the relations $e_X \affects YW \given \doo(Z X \setminus e_X)$, implied for all $e_X \in X$ by Irred$_1$, similarly imply $e_X \affects YW \ \forall e_X \in X$.
    By compatibility, this implies $\Fut(e_\XX) \supseteq \Fut_s(\YY\WW) \ \forall e_X \in X$, which we can equivalently write as $\Fut_s(\XX) \supseteq \Fut_s(\YY\WW)$.
    Hence, the statement follows with $s_Z = \emptyset$.
\end{proof}

\begin{theorem}
    For any spacetime, all compatible embeddings of ACLs rooted in a causal model without clustering are trivial.
\end{theorem}

\begin{proof}
    Consider a set of affects relations which imply an ACL.
    Then by \cref{thm:deorder-loop}, we can transform it into a set of unconditional 0$^\text{th}$-order affects relations, with the same implications for causal inference and compatibility, and continue with this transformed set.

    Choose an ORV $e_\XX$ that is part of a causal loop, as inferrable from the presence of a set of irreducible affects relations.
    Then, it must be part of at least one irreducible affects relation $X \affects Y$ with $X \ni e_X$, implying $e_X \cause Y$ with $Y \subset S \setminus e_X$.
    Due to \cref{thm:no-clus-conditional}, in absence of clustering each such affects relation can be resolved such that $e_X \affects e_Y$ for some $e_Y \in Y$, yielding the compatibility condition $e_\XX \preceq e_\YY$.
    Iterating, if $e_Y$ is part of the causal loop, it again admits some $e_Z$ such that $e_Y \affects e_Z$ and $e_\YY \preceq e_\ZZ$.
    Being restricted to unconditional $0^\text{th}$-order affects relations atomic both in their first and second argument, the resulting ACL must be a loop of single variables
    $e_\XX \preceq e_\YY \preceq e_\ZZ \preceq \ldots \preceq e_\XX$ (i.e.\ an ACL type 2 by~\cite{VVC}), with all variables involved collapsing into a single point in spacetime.
    Therefore, we obtain a degenerate embedding, which is furthermore \emph{trivial} according to the definitions of~\cite{VVC}, as multiple ORVs individually signalling to one another share the same location.
	(See also Lemma VI.2 of~\cite{VVC}.)
\end{proof}

\section{A different notion of irreducibility?}
\label{sec:irreducible-weak}
\enlargethispage*{3\baselineskip}

Considering the proof of \cref{thm:deorder}, it is natural to wonder whether one should consider a weaker notion of irreducibility (and dually, a stronger notion of reducibility) than the ones studied in \cite{VVC,Grothus2024}, to obtain a more well-behaved interplay between different arguments.

Specifically, for the first argument of an affects relation $X \affects Y$, we would require only $e_X \affects Y \given \doo(Z X \setminus e_X)$ for all $e_X \in X$ with respect to the first argument, while for the third argument, we would require $e_Z \affects Y \given \doo( X Z \setminus e_Z) \ \lor \ e_Z \affects Y \given \doo(Z \setminus e_Z)$ for all $e_Z \in Z$.
By contrast, the present definition of irreducibility uses \textit{for all $s_{X}\subset X$} and \textit{for all $s_{Z}\subset Z$} in place of $e_X$ and $e_Z$ respectively.

Due to these definitions, affects relations which are only weakly irreducible have a cardinality of $\abs{X} = 3$ or $\abs{Z} = 2$ at least, to respectively be weakly irreducible in this sense without being irreducible in the usual sense.
However, no example for a causal model which gives rise to an affects relation which is only weakly irreducible in some argument has been established so far.

Notably, this weaker notion is sufficient to reproduce the causal inference results for \enquote{irreducible} affects relations that have been established in Section~4.3 of \cite{Grothus2024}.
Aided by the absence of clustering in the third argument or conicality of the embedding, it is also sufficient to reproduce all compatibility results (cf.\ Appendix~F of \cite{Grothus2024}).
Indeed, the weaker definition is therefore sufficient to reproduce the entirety of our main results (disregarding \cref{sec:clus}).
In the non-conical case however, we can not rule out that replacing irreducibility in the definition of compatibility (cf.\ \cref{sec:overview}) with this weaker notion may pose a stronger restriction on embeddings.

However, we do not choose to generally adopt this notion, as this weaker notion actually implies some characteristics one would usually intuitively designate as reducible:
If an affects relation $X \affects Y$ is only weakly irreducible, for some $\tilde{s}_X \subsetneq X$, some \emph{reduced} affects relation $\tilde{s}_X \affects Y$ must hold as well.
\begin{lemma}
    Consider an affects relation $X \affects Y \given \doo(Z)$, which is Red$_1$, yet satisfies $e_X \affects Y \given \doo(Z X \setminus e_X)$ for all $e_X \in X$ and is hence weakly irreducible.
    Then there exists $s_X \subsetneq X$ with $\abs{s_X} \ge 2$ such that $X \setminus s_X \affects Y \given \doo(Z)$.
\end{lemma}
\begin{proof}
    As $X \affects Y \given \doo(Z)$ is Red$_1$, there exists a non-empty subset $s_X\subsetneq X$ such that $s_X \naffects Y | \mathrm{do}(\tilde{s}_XZ)$, where $\tilde{s}_X:=X\backslash s_X$, with $\abs{s_X} > 2$.
    However, jointly with the original relation, this implies the \mbox{(in-)equality} chain
    $P_{\cG_{\doo(Z)}} (Y|Z) \neq P_{\cG_{\doo(XZ)}} (Y|XZ) = P_{\cG_{\doo(\tilde{s}_XZ)}} (Y|\tilde{s}_XZ)$,
    and hence, $\tilde{s}_X \affects Y \given \doo(Z)$.
\end{proof}
This matches the fact that as shown in \cite{Grothus2024}, no affects relation which is not reducible (in the usual sense) can be clustered.

One may also consider an analogously weakened property for clustering.
However, this property does not admit an as nice interpretation as the original version, which encodes the inability to signal when considering only a subset of RVs in some argument.
This e.g.\ naturally captures the properties of secret sharing protocols \cite{Cleve1999, PhysRevA.59.1829,PhysRevA.59.162,PhysRevA.61.042311,Elliot}, where information encoded in multiple systems cannot be recovered from subsets thereof.
Specifically, the weakened property would require only $X \setminus e_X \naffects Y \given \doo(Z)$ for all $e_X \in X$ for the first argument, rather than $s_X \naffects Y \given \doo(Z)$ for all $s_X \in Y$,
and similarly $X \naffects Y \given \doo(Z \setminus e_Z)$ for all $e_Z \in Z$ for the third argument.
Similarly, this weakened notion would be sufficient to reproduce our causal inference results of \cite{Grothus2024} for clustered affects relations.
We leave a deeper study of these weaker notions for future work, in particular the interpretation and implications of this weaker notion of clustering for compatible embeddings and causal loops.
\end{document}